\documentclass[11pt]{article}%
\usepackage{amsmath}
\usepackage{amsfonts}
\usepackage{amssymb}
\usepackage{graphicx}%
\providecommand{\U}[1]{\protect\rule{.1in}{.1in}}
\newtheorem{theorem}{Theorem}
\newtheorem{acknowledgement}[theorem]{Acknowledgement}

\newtheorem{lemma}[theorem]{Lemma}
\newtheorem{notation}[theorem]{Notation}

\newtheorem{proposition}[theorem]{Proposition}

\newenvironment{proof}[1][Proof]{\noindent\textbf{#1.} }{\ \rule{0.5em}{0.5em}}
\begin{document}

\title{Geometric Representation of Generalized Coherent States and their Symplectic
Capacities: A Synthetic Approach}
\author{Maurice de Gosson\\Austrian Academy of Sciences\\Acoustics Research Institute\\1010, Vienna, AUSTRIA\\and\\University of Vienna\\Faculty of Mathematics (NuHAG)\\1090 Vienna, AUSTRIA}
\maketitle

\begin{abstract}
In this work we review, complete, and synthesize results linking generalized
coherent stages (nondegradable Gaussian wavefunctions) to the notions of Fermi
ellipsoids, quantum blobs, and microlocal pairs introduced in previous work.
These geometric objects are Fermi ellipsoids, quantum blobs, and microlocal
pairs. In addition we study various symplectic capacities associated with
these objects.

\end{abstract}

\section{Introduction}

The notion of coherent states was introduced by Erwin Schr\"{o}dinger
\cite{Schr1} as minimum uncertainty wavepackets of the form
\begin{equation}
\psi(x)=A\exp\left(  -\frac{(x-x_{0})^{2}}{2\sigma^{2}}+i\frac{p_{0}x}{\hbar
}\right)  . \label{coh1}%
\end{equation}
What Schr\"{o}dinger had in mind was to study the states of the quantum
harmonic oscillator which minimize the Heisenberg uncertainty relation $\Delta
x\cdot\Delta p\geq\frac{1}{2}\hbar\ $ The notion of coherent states has been
since then considerable extended by and studied many several authors,\textit{
e.g}. \cite{1,2,3}. In the present paper we will consider generalized
non-degenerate Gaussians of the type
\begin{equation}
\psi_{X,Y}(x)=\left(  \tfrac{\det X}{(\pi\hbar)^{n}}\right)  ^{1/4}%
e^{-\frac{1}{2\hbar}(X+iY)x\cdot x} \label{coh2}%
\end{equation}
where $X,Y\in\operatorname*{Sym}(n,\mathbb{R})$ ,and $X$ is positive definite
($X>0$) as well as there extension
\begin{equation}
\psi_{X,Y}^{x_{0},p_{0}}(x)=\left(  \tfrac{\det X}{(\pi\hbar)^{n}}\right)
^{1/4}e^{i\frac{p_{0}x}{\hbar}}e^{-\frac{1}{2\hbar}(X+iY)(x-x_{0}%
)\cdot(x-x_{0})}. \label{coh3}%
\end{equation}
We will call such states "generalized coherent states". Functions of this type
are widely in quantum optics (where they have become an industry), in quantum
mechanics, and also time-frequency analysis where they are omnipresent in the
theory of Gabor frames.

We are going to show that generalized Gaussianity wavepackets (or states) can
be represented in three different ways by geometrical constructs:

\begin{itemize}
\item \textit{Fermi ellipsoids}, which are the interior of the level sets
determined by the Weyl symbol stationary second-order partial defensible
equitation
\[
\left[  (-i\hbar\nabla_{x}+Yx)^{2}+X^{2}x\cdot x\right]  \psi_{X,Y}%
=(\hbar\operatorname*{Tr}X)\psi_{X,Y}%
\]
satisfied by $\psi_{X,Y}$; it is the Fermi ellipsoid is thus the ellipsoid
defined by the inequality
\[
(p+Yx)^{2}+X^{2}x\cdot x\leq\hbar\operatorname*{Tr}X.
\]
This construction is easily extended to the case of translated Gaussians
(\ref{coh3});

\item \textit{Quantum blobs}, which are symplectic balls with radius
$\sqrt{\hbar}$; they are minimum uncertainty phase space cells and are
represented as the Wigner covariance ellipsoid of $\psi_{X,Y}$. The
identification of quantum blobs with generalized coherent states is related to
the uncertainty principle via its conventional formulation in terms of
variances and covariances, and is thus statistical in nature.. In this sense
this identification is not intrinsic, as is that proposed by Fermi ellipsoids;

\item \textit{Microlocal pairs}; we (introduced in \cite{Poiuntilliswme}): we
are using Fefferman's terminology \cite{Fefferman} to qualify Cartesian
products $X_{\ell}\times X_{\ell^{\prime}}^{\hbar}$ where $X_{\ell}$ is an
ellipsoid carried by a Lagrangian plane $\ell$ in phase space and
$X_{\ell^{\prime}}^{\hbar}$ is the $\hbar$-polar dual of $X_{\ell}$ with
respect to a transverse Lagrangian plane $\ell^{\prime}$ . It turns out that
the John ellipsoid of a microlocal pair is always a quantum blob.
\end{itemize}

Denoting by $\operatorname{Fermi}(n)$, $\operatorname*{Blob}(n)$,
$\operatorname{Micro}(n)$ the sets described above and by
$\operatorname*{Gauss}(n)$ the of all generalized coherent states (\ref{coh3})
we will prove the existence of three bijections
\begin{align*}
\Phi_{\operatorname{Fermi}}  &  :\operatorname{Fermi}(n)\longrightarrow
\operatorname*{Gauss}(n)\\
\Phi_{\mathrm{blob}}  &  :\operatorname*{Blob}(n)\longrightarrow
\operatorname*{Gauss}(n)\\
\Phi_{\mathrm{micro}}  &  :\operatorname{Micro}(n)\longrightarrow
\operatorname*{Gauss}(n)
\end{align*}
These bijections thus provide three different, but equivalent, geometric
representations of generalized coherent states. Being geometric objects in
phase space it makes sense to study their topological properties. This will be
achieved by calculating the symplectic capacities of the objects in
$\operatorname{Fermi}(n)$, $\operatorname*{Blob}(n)$, and
$\operatorname{Micro}(n$

The basic properties of the symplectic group and its metaplectic
representation are given in Appendix A; the main definitions and properties of
symplectic capacities are given in Appendix B.

\begin{notation}
The points in configuration and momentum space are written $x=(x_{1}%
,...,x_{n})$ and $p=(p_{1},...,p_{n})$ respectively; in formulas $x$ an $p$
are viewed as column vectors. We will also use the collective notation
$z=(x,p)$ for the phase space variable.
\end{notation}

\section{Generalized coherent states}

We quickly review the mathematical objects which are the main theme of this paper.

\subsection{Description and notation}

The generalized Gaussians of the type
\begin{equation}
\psi_{X,Y}(x)=\left(  \tfrac{\det X}{(\pi\hbar)^{n}}\right)  ^{1/4}%
e^{-\frac{1}{2\hbar}(X+iY)x\cdot x} \label{squeezed}%
\end{equation}
($X,Y\in\operatorname*{Sym}(n,\mathbb{R})$ and$X>0$) can be obtained from the
standard coherent state
\begin{equation}
\phi_{0}(x)=\psi_{I,0}(x)=(\pi\hbar)^{-n/4}e^{-|x|^{2}/2\hbar}
\label{standard}%
\end{equation}
using the elementary unitary operators, which belong to the metaplectic group
$\operatorname*{Mp}(n)$ (see Appendix A),
\begin{equation}
\widehat{V}_{P}\psi(x)=e^{--\frac{i}{2\hbar}Px\cdot x}\psi(x)\text{ \ ,
\ }\widehat{M}_{L}\psi(x)=\sqrt{\det L}\psi(XL) \label{unit}%
\end{equation}
as follows from the obvious formula
\begin{equation}
\psi_{X,Y}=\widehat{V}_{Y}\widehat{M}_{X^{1/2}}\phi_{0}. \label{fixy}%
\end{equation}

More generally, we define the displaced Gaussians
\begin{equation}
\psi_{X,Y}^{z_{0}}=\widehat{T}(z_{0})\psi_{X,Y}\text{ \ \ },\text{ \ \ }%
z_{0}=(x_{0},p_{0}) \label{psitt}%
\end{equation}
where $\widehat{T}(z_{0})$ is the Heisenberg--Weyl operator%
\begin{equation}
\widehat{T}(z_{0})\psi(x)=e^{\frac{i}{\hbar}(p_{0}\cdot x-\frac{1}{2}%
p_{0}\cdot x_{0})}\psi(x-x_{0})) \label{HW}%
\end{equation}
(note that this definition slightly differs from the one given in formula
(\ref{coh2}) in the introduction; the one given here is more flexible because
of its covariance properties as we will see).

We denote by $\operatorname*{Gauss}(n)$ the set of all equivalence classes
$|\psi_{X,Y}^{z_{0}}\rangle$ \ of Gaussian states\ for the equivalence
relation
\[
\psi_{X,Y}^{z_{0}}\sim\psi_{X^{\prime},Y^{\prime}}^{z_{0}^{\prime}%
}\Leftrightarrow\psi_{X^{\prime},Y^{\prime}}^{z_{0}^{\prime}}=c\psi
_{X,Y}^{z_{0}}\text{ \ , \ }|c|=1.
\]
The subset of $\operatorname*{Gauss}(n)$ consisting of all centered Gaussians
($z_{0}=0$). is denoted by by $\operatorname*{Gauss}_{0}(n)$.

\begin{proposition}
The metaplectic group $\operatorname*{Mp}(n)$ acts transitively on
$\operatorname*{Gauss}(n)$: for every pair $|\psi_{X,Y}^{z_{0}}\rangle$,
$|\psi_{X^{\prime},Y^{\prime}}^{z_{0}^{\prime}}\rangle$ in
$\operatorname*{Gauss}_{0}(n)$ there exists $\hat{S}\in\operatorname*{Mp}(n)$
such that $\hat{S}|\psi_{X^{\prime},Y^{\prime}}^{z_{0}^{\prime}}\rangle
=|\psi_{X,Y}^{z_{0}}\rangle$.
\end{proposition}

\begin{proof}
It immediately follows from formula \ (\ref{fixy}).
\end{proof}

We leave it to the reader to extend the result above to the case of
$\operatorname*{Gauss}(n)$ by using the extended metaplectic group obtained by
forming the semi direct product of $\operatorname*{Mp}(n)$ with the
Heisenberg--Weyl operators \cite{Folland,Birk,Birkbis}.

\subsection{The Wigner function of $\psi_{X,Y}$}

The Wigner transform
\begin{equation}
W\psi_{X,Y}(z)=\left(  \frac{1}{2\pi\hbar}\right)  ^{n}\int_{\mathbb{R}^{n}%
}e^{-\frac{i}{\hbar}p\cdot y}\psi_{X,Y}(x+\tfrac{1}{2}y)\psi_{X,Y}^{\ast
}(x-\tfrac{1}{2}y)dy \label{oupsi}%
\end{equation}
of the Gaussian state $\psi_{X,Y}$ is itself a Gaussian, namely%
\begin{equation}
W\psi_{X,Y}(z)=\left(  \frac{1}{\pi\hbar}\right)  ^{n}\exp\left(  -\frac
{1}{\hbar}Gz\cdot z\right)  \label{goupsi}%
\end{equation}
where $G$ is the matrix%
\begin{equation}
G=%
\begin{pmatrix}
X+YX^{-1}Y & YX^{-1}\\
X^{-1}Y & X^{-1}%
\end{pmatrix}
\label{G}%
\end{equation}
(see e.g. \cite{Birk,Littlejohn}). It is a positive definite symplectic
matrix, for we can write
\begin{equation}
G=S^{T}S\text{ \ , \ }S=%
\begin{pmatrix}
X^{1/2} & 0\\
X^{-1/2}Y & X^{-1/2}%
\end{pmatrix}
\label{ess}%
\end{equation}
and $S$ is trivially symplectic. Setting $G=\frac{\hbar}{2}\Sigma^{-1}$ and
using the fact that $\det G=1$ we can rewrite (\ref{goupsi}) as%
\begin{equation}
W\psi_{X,Y}(z)=\left(  \frac{1}{2\pi}\right)  ^{n}\det\Sigma^{-1/2}\exp\left(
-\frac{1}{2}\Sigma^{-1}z\cdot z\right)
\end{equation}
so $\Sigma^{-1}$ is the covariance matrix of the phase space Gaussian
distribution $W\psi_{X,Y}$. Considering the corresponding covariance ellipsoid
(summertimes called the Wigner ellipsoid \cite{Littlejohn})%

\begin{equation}
\Omega_{\Sigma}=\{z:\frac{1}{2}\Sigma^{1}z\cdot z\leq1\}=\{z:\frac{1}%
{2}Gz\cdot z\leq\hbar\} \label{cowig}%
\end{equation}
we see that $\Omega_{\Sigma}$ is a symplectic ball with radius$\sqrt{\hbar}$:%
\begin{equation}
\Omega_{\Sigma}=S^{-1}(B^{2n}(\sqrt{\hbar}) \label{omegablob}%
\end{equation}
that is a \textit{quantum blob} \textbf{\cite{blob}. }Notice that
\begin{equation}
S^{-1}=%
\begin{pmatrix}
X^{-1/2} & 0\\
-X^{-1/2}Y & X^{1/2}%
\end{pmatrix}
=M_{X^{1/2}}V_{Y}. \label{smin1}%
\end{equation}

Similarly, using the covariance property relation
\[
W(\widehat{T}(z_{0})\psi)=W\psi(z-z_{0})
\]
Wigner transform and Heisenberg--Weyl operators we have
\[
W\psi_{X,Y}^{z_{0}}(z)=\left(  \frac{1}{\pi\hbar}\right)  ^{n}\exp\left(
-\frac{1}{\hbar}G(z-z_{0})\cdot(z-z_{0})\right)
\]
where the matrix $G$ is defined as above. The extension of the discussion
above to the translated case is straightforward.

\section{The Fermi Representation of $\psi_{X,Y}$}

\subsection{The stationary equation satisfied by $\psi_{X,Y}$}

Enrico Fermi in a largely forgotten paper \cite{Fermi} from 1930. Fermi
associates to every quantum state $\psi$ a certain hypersurface $g_{\mathrm{F}%
}(x,p)=0$ in phase space. Consider a complex twice continuously differentiable
function $\psi(x)=R(x)e^{i\phi(x)/\hslash}$ ($R(x)\geq0$ and $\phi(x)$ real)
defined on configuration space $\mathbb{R}^{n}$. At every $x$ where
$R(x)\neq0$\ the function $R$ trivially satisfies the identity
\begin{equation}
\left(  -\hbar^{2}\nabla_{x}^{2}+\hbar^{2}\frac{\nabla_{x}^{2}R(x)}%
{R(x)}\right)  R(x)=0. \label{fermi1}%
\end{equation}
Performing the gauge transformation $-i\hbar\nabla_{x}\longrightarrow
-i\hbar\nabla_{x}-\nabla_{x}\phi$ this identity becomes
\begin{equation}
\left(  -i\hbar\nabla_{x}-\nabla_{x}\phi(x)\right)  ^{2}-Q(x))R(x)=0
\label{16}%
\end{equation}
where $Q$ is the real function%
\begin{equation}
Q(x)=--\hbar^{2}\frac{\nabla_{x}^{2}R(x\mathbf{)}}{R\mathbf{(}x\mathbf{)}}%
\end{equation}
(the reader familiar with Bohm's approach to quantum mechanics will note that
$Q(x)$ can be identified with Bohm's quantum potential; up to a factor the
idea has been developed in \cite{DeGoHi}). We will call
\begin{equation}
\left(  \left(  -i\hbar\nabla_{x}-\nabla_{x}\phi\right)  ^{2}+\hbar^{2}%
\frac{\nabla_{x}^{2}R}{R}\right)  \psi_{X,Y}=0 \label{fermiop}%
\end{equation}
(the procedure outlined above goes back to Fermi \cite{Fermi}; also see
\cite{best}).

Let us apply this to the Gaussian $\psi_{X,Y}$.

\begin{proposition}
The generalized coherent state $\psi_{X,Y}$ satisfies the eigenvalue equation%
\begin{equation}
\widehat{H}_{XY}\psi_{X,Y}=(\hbar\operatorname*{Tr}X)\psi_{X,Y} \label{stat}%
\end{equation}
where $\widehat{H}_{\mathrm{F}}$ is the partial differential operator
\begin{equation}
\widehat{H}_{XY}=(-i\hbar\nabla_{x}+Yx)^{2}+X^{2}x\cdot x. \label{pdo}%
\end{equation}

\end{proposition}

\begin{proof}
Setting $\phi(x)=-\frac{1}{2}Yx\cdot x$ and $R(x)=\exp\left(  -\frac{1}%
{2\hbar}Xx\cdot x\right)  $ we have
\begin{equation}
\nabla_{x}\phi(x)=-Yx\text{ \ , \ }\frac{\nabla_{x}^{2}R(x)}{R(x)}=-\frac
{1}{\hbar}\operatorname*{Tr}X+\frac{1}{\hbar^{2}}X^{2}x\cdot x \label{tr}%
\end{equation}
hence, taking (\ref{fermio}) into account, $\psi_{X,Y}$ is a solution of
(\ref{stat}).
\end{proof}

We call $\widehat{H}_{XY}$ the Fermi operator of $\psi_{X,Y}$.

Choosing for $\psi_{X,Y}$ the standard coherent state $\phi_{0}(x)=(\pi
\hbar)^{-n/4}e^{-|x|/2\hbar}$ $we$ have $X=I$ and $Y=0$ hence the equation
(\ref{stat}) reduces, dividing by $2$, to the familiar harmonic oscillator
eigenvalue equation%
\[
\frac{1}{2}(-\hbar^{2}\nabla_{x}^{2}+X^{2})\phi_{0}=\frac{1}{2}n\hbar\phi
_{0}.
\]

\subsection{The Fermi ellipsoid and generalized coherent states}

The Fermi operator$\widehat{H}_{\mathrm{F}}$ is the quantized version (in any
reasonable quantization scheme) of the function
\begin{equation}
H_{XY}(x,p)=(p+Yx)^{2}+X^{2}x\cdot x-\hbar\operatorname*{Tr}X \label{gf3}%
\end{equation}
(it is the Weyl symbol of $\widehat{H}_{\mathrm{F}}$). We can rewrite this
formula as
\[
H_{XY}(x,p)=M_{XY}z\cdot z-\hbar\operatorname*{Tr}X
\]
($z=(x,p)$) where $M_{XY}$ is the symmetric matrix
\begin{equation}
M_{XY}=%
\begin{pmatrix}
X^{2}+Y^{2} & Y\\
Y & I
\end{pmatrix}
. \label{mf}%
\end{equation}
A straightforward calculation shows that this matrix admits the factorization%
\begin{equation}
M_{XY}=S^{T}D_{X}S\text{ \ , \ }D_{X}=%
\begin{pmatrix}
X & 0\\
0 & X
\end{pmatrix}
\label{mfs}%
\end{equation}
where $S$ is the symplectic matrix in (\ref{ess}) \
\begin{equation}
S=%
\begin{pmatrix}
X^{1/2} & 0\\
X^{-1/2}Y & X^{-1/2}%
\end{pmatrix}
. \label{sxy}%
\end{equation}

It turns out --and this is really a striking fact!-- that $M_{XY}$ is closely
related to the Wigner transform (\ref{oupsi})--(\ref{goupsi}) of $\psi_{X,Y}$.
In fact%
\begin{equation}
W\psi_{X,Y}(z)=\left(  \frac{1}{\pi\hbar}\right)  ^{n}e^{-\operatorname*{Tr}%
X}\exp\left(  -\frac{1}{\hbar}H_{XY}(S^{-1}D^{-1/2}Sz)\right)  \label{wgf}%
\end{equation}
with $D=%
\begin{pmatrix}
X & 0\\
0 & X
\end{pmatrix}
$. In particular, when $n=1$ and $\psi_{X,Y}(x)=\psi_{0}(x)$ the standard
coherent state we have%
\[
W\psi_{0}(z)=\left(  \frac{1}{\pi\hbar}\right)  ^{1/4}\frac{1}{e}\exp\left(
-\frac{1}{\hbar}M_{XY}z\cdot z\right)
\]
which was already observed in \cite{best}.

We will, as Fermi did, identify the state $\psi_{X,Y}$ with the phase space
ellipsoid defined by $H_{XY}(x,p)\leq0,$ that is:%
\begin{equation}
\Omega_{XY}=\{z:M_{XY}z\cdot z\leq\hbar\operatorname*{Tr}X\}. \label{Fermi1}%
\end{equation}
We will call $\Omega_{XY}$ the \textit{Fermi ellipsoid} of the generalized
coherent state $\psi_{X,Y})$.

We denote by $\operatorname{Fermi}(n)$ the set of all Fermi ellipsoids%
\begin{equation}
\Omega_{XY}^{z_{0}}=\{z:M_{XY}(z-z_{0})\cdot(z-z_{0})\leq\hbar
\operatorname*{Tr}X\}. \label{Fermi2}%
\end{equation}
and by $\operatorname{Fermi}_{0}(n)$ the subset consisting of all centered
Fermi ellipsoids (\ref{Fermi1}).

\begin{proposition}
The mapping $\Phi_{\operatorname{Fermi}}:\operatorname{Fermi}_{0}%
(n)\longrightarrow\operatorname*{Gauss}_{0}(n)$ defined by%
\[
\Phi_{\operatorname{Fermi}}:\Omega_{XY}\longmapsto\psi_{X,Y}%
\]
is a bijection.
\end{proposition}

\begin{proof}
The ellipsoid $\Omega_{XY}$ uniquely determined the matrices $X$ and $Y$ via
its matrix $M_{XY}$ (\ref{mf}). \ If $M_{XY}\neq M_{X^{\prime}Y^{\prime}%
}^{\prime}$ we cannot have both $X=X^{\prime}$ and $Y=Y^{\prime}$ hence
$\psi_{X,Y}\neq\psi_{X^{\prime},Y^{\prime}}$. Surjectivity is obvious.
\end{proof}

One constructs similarly a bijection $\operatorname{Fermi}(n)\longrightarrow
\operatorname*{Gauss}(n)$ by translating Fermi ellipsoids and applying the
Heisenberg--Weyl operators to the Gaussians $\psi_{X,Y}$.

\subsection{Symplectic capacities of Fermi ellipsoids}

We denote by $\Omega_{XY}$ the ellipsoid\ $M_{\mathrm{F}}z\cdot z\leq
\hbar\operatorname*{Tr}X$ bounded by the Fermi hypersurface $\mathcal{H}%
_{\mathrm{F}}$ corresponding to the squeezed coherent state $\psi_{X,Y}$. Let
us perform the symplectic change of variables $z^{\prime}=Sz$; in the new
coordinates the ellipsoid $\Omega_{XY}$ is represented by the inequality
\begin{equation}
Xx^{\prime}\cdot x^{\prime}+Xp^{\prime}\cdot p^{\prime}\leq\hbar
\operatorname*{Tr}X \label{ferx}%
\end{equation}
hence $c(\Omega_{XY})$ equals the symplectic capacity of the ellipsoid
(\ref{ferx}). Applying the rule above we thus have to find the symplectic
eigenvalues of the block-diagonal matrix $%
\begin{pmatrix}
X & 0\\
0 & X
\end{pmatrix}
$; a straightforward calculation shows that these are just the eigenvalues
$\omega_{1},...,\omega_{n}$ of $X$ and hence%
\begin{equation}
c(\Omega_{XY})=\pi\hbar\operatorname*{Tr}X/\omega_{\max} \label{cwf}%
\end{equation}
where $\omega_{\max}=\max\{\omega_{1},...,\omega_{n}\}$. (Cf. the proof of
Proposition \ref{propinclus}). In view of the trivial inequality
\begin{equation}
\omega_{\max}\leq\operatorname*{Tr}X=\sum_{j=1}^{n}\omega_{j}\leq
n\lambda\omega_{\max} \label{maxnmax}%
\end{equation}
it follows that we have
\begin{equation}
\frac{1}{2}h\leq c(\Omega_{XY})\leq\frac{nh}{2}. \label{nh}%
\end{equation}
Notice that when all the eigenvalues $\omega_{j}$ are equal to a number
$\omega$ then $c(\Omega_{XY})=nh/2$; in particular when $n=1$ we have
$c(\Omega_{XY})=h/2$ which is exactly the action calculated along the
trajectory corresponding to the ground state. Let us come back to the
interpretation of the ellipsoid defined by the inequality (\ref{ferx}). We
have seen that the symplectic eigenvalues of the matrix $%
\begin{pmatrix}
X & 0\\
0 & X
\end{pmatrix}
$ are precisely the eigenvalues $\omega_{j}$, $1\leq j\leq n$, of the
positive-definite matrix $X$. It follows that there exist symplectic
coordinates $(x^{\prime\prime},p^{\prime\prime})$ in which the equation of the
ellipsoid $\Omega_{XY}$ takes the normal form%
\begin{equation}
\sum_{j=1}^{n}\omega_{j}(x_{j}^{\prime\prime2}+p_{j}^{\prime\prime2})\leq
\sum_{j=1}^{n}\hbar\omega_{j} \label{omf1}%
\end{equation}
whose quantum-mechanical interpretation is clear: dividing both sides by two
we get the energy shell of the anisotropic harmonic oscillator in its ground
state. Consider now the planes $\mathcal{P}_{1},\mathcal{P}_{2},..,\mathcal{P}%
_{n}$ of conjugate coordinates $(x_{1},p_{1})$, $(x_{2},p_{2})$,...,
$(x_{n},p_{n})$. The intersection of the ellipsoid $\Omega_{XY}$ with these
planes are the circles
\begin{gather*}
C_{1}:\omega_{1}(x_{1}^{\prime\prime2}+p_{1}^{\prime\prime2})\leq\sum
_{j=1}^{n}\hbar\omega_{j}\\
C_{2}:\omega_{2}(x_{2}^{\prime\prime2}+p_{2}^{\prime\prime2})\leq\sum
_{j=1}^{n}\hbar\omega_{j}\\
\cdot\cdot\cdot\cdot\cdot\cdot\cdot\cdot\cdot\cdot\cdot\cdot\cdot\cdot
\cdot\cdot\cdot\\
C_{n}:\omega_{n}(x_{n}^{\prime\prime2}+p_{n}^{\prime\prime2})\leq\sum
_{j=1}^{n}\hbar\omega_{j}.
\end{gather*}
Formula (\ref{cwf}) says that $c(\Omega_{XY})$ is precisely the area of the
circle $C_{j}$ with smallest area., which corresponds to the index $j$ such
that $\omega_{j}=\omega_{\max}$. This is of course perfectly in accordance
with the definition of the Hofer--Zehnder capacity $c^{_{\mathrm{HZ}}}%
(\Omega_{XY})$ since all symplectic capacities agree on ellipsoids. This leads
us now to another question: is there any way to describe topologically fermi's
ellipsoid in such a way that the areas of every circle $C_{j}$ becomes
apparent? The problem with the standard capacity of an ellipsoid is that it
only \textquotedblleft sees\textquotedblright\ the smallest cut of that
ellipsoid by a plane of conjugate coordinate. The way out of this difficult
lies in the use of the Ekeland--Hofer capacities $c_{j}^{\mathrm{EH}}$
described in the Appendix. To illustrate the idea, let us first consider the
case $n=2$; it is no restriction to assume $\omega_{1}\leq\omega_{2}$. If
$\omega_{1}=\omega_{2}$ then the ellipsoid%
\begin{equation}
\omega_{1}(x_{1}^{\prime\prime2}+p_{1}^{\prime\prime2})+\omega_{2}%
(x_{2}^{\prime\prime2}+p_{2}^{\prime\prime2})\leq\hbar\omega_{1}+\hbar
\omega_{2} \label{omf2}%
\end{equation}
is the ball $B^{2}(\sqrt{2\hbar})$ whose symplectic capacity is $2\pi\hbar=h$.
Suppose now $\omega_{1}<\omega_{2}$. Then the Ekeland--Hofer capacities are
the numbers%
\begin{equation}
\frac{\pi\hbar}{\omega_{2}}(\omega_{1}+\omega_{2}),\frac{\pi\hbar}{\omega_{1}%
}(\omega_{1}+\omega_{2}),\frac{2\pi\hbar}{\omega_{2}}(\omega_{1}+\omega
_{2}),\frac{2\pi\hbar}{\omega_{1}}(\omega_{1}+\omega_{2}),.... \label{seq}%
\end{equation}
and hence
\[
c_{1}^{\mathrm{EH}}(\Omega_{XY})=c(\Omega_{XY})=\frac{\pi\hbar}{\omega_{2}%
}(\omega_{1}+\omega_{2}).
\]
What about $c_{2}^{\mathrm{EH}}(\Omega_{XY})$? A first glance at the sequence
(\ref{seq}) suggests that we have
\[
c_{2}^{\mathrm{EH}}(\Omega_{XY})=\frac{\pi\hbar}{\omega_{1}}(\omega_{1}%
+\omega_{2})
\]
but this is only true if $\omega_{1}<\omega_{2}\leq2\omega_{1}$ because if
$2\omega_{1}<\omega_{2}$ then $(\omega_{1}+\omega_{2})/\omega_{2}<(\omega
_{1}+\omega_{2})/\omega_{1}$ so that in this case
\[
c_{2}^{\mathrm{EH}}(\Omega_{XY})=\frac{\pi\hbar}{\omega_{2}}(\omega_{1}%
+\omega_{2})=c_{1}^{\mathrm{EH}}(\Omega_{XY}).
\]
The Ekeland--Hofer capacities thus allow a topological classification of the eigenstate.

\section{Quantum Blobs and $\psi_{X,Y}$}

We have introduced in our previous work \cite{blob,goluPR} the notion of
"quantum blob" as being a phase space cell of minimum uncertainty. Rigorously
stated a quantum blob is the image of a phase space ball with radius
$\sqrt{\hbar}by$ a linear symplectic transformation. It follows from Gromov's
non-squeezing theorem \cite{gr85} that the shadow of a quantum blob on any
plane of conjugate coordinates $x_{j},p_{j}$ is always at least $\pi
\hbar=\frac{1}{2}h$ whereas its shadow on planes of non-conjugate coordinates
can be arbitrarily small.

\subsection{Definition and properties}

By definition a quantum blob is the image of the phase space ball
$B^{2n}(z_{0},\sqrt{\hbar}):|z-z_{0}|\leq\sqrt{\hbar}$ by a linear canonical
transformation $S\in\operatorname*{Sp}(n)$. We denote by $\operatorname*{Blob}%
(n)$ the set of all quantum blobs in $\mathbb{R}^{2n}$ and by
$\operatorname*{Blob}_{0}(n)$ the subset consisting of all quantum blobs
centered at $z_{0}=0$, that is $S(B^{2n}\sqrt{\hbar}))$ where $B^{2n}%
\sqrt{\hbar})$ is the ball with radius $^{2n}\sqrt{\hbar}$ centered at the
origin. We have
\[
S(B^{2n}(z_{0},\sqrt{\hbar}))=S(T(z_{0})B^{2n}\sqrt{\hbar}))=T(Sz_{0}%
)S(B^{2n}\sqrt{\hbar}))
\]
showing that every element of $\operatorname*{Blob}(n)$ is obtained by
translating an element of $\operatorname*{Blob}_{0}(n).$

The following important result that quantum blobs do not use the full
symplectic group to be defined:

\begin{proposition}
For every granum blob $S(B^{2n}(z_{0},\sqrt{\hbar}))\in\operatorname*{Blob}%
(n)$ there exist unique matrices $P^{T}$ and $L=L^{T}>0$ such that
\begin{equation}
S(B^{2n}(z_{0},\sqrt{\hbar}))=V_{P}M_{L}(B^{2n}(z_{0},\sqrt{\hbar}))).
\label{SB}%
\end{equation}
If $S=%
\begin{pmatrix}
A & B\\
C & D
\end{pmatrix}
$ then
\begin{gather}
L=(AA^{T}+BB^{T})^{-1/2}=L^{T}>0\label{pl1}\\
P=-(CA^{T}+DB^{T})(AA^{T}+BB^{T})^{-1}=P^{T}. \label{pl2}%
\end{gather}

\end{proposition}

\begin{proof}
It suffices to consider the case $z_{0}=0$. Recall the pre-Iwasawa (KAM)
decomposition \cite{Dutta,iwa} (see \cite{Birk}, Ch.2, for detailed
calculations). of a symplectic matrix. It says that for every $S\in
\operatorname*{Sp}(n)$there exist unique matrices $L$ and $P$ symmetric,
$X>0$, and $R\in\operatorname*{Sp}(n)\cap O(2n)$ such that $S=V_{P}M_{L}R$
\textit{i.e.}
\begin{equation}
S=%
\begin{pmatrix}
I & 0\\
P & I
\end{pmatrix}%
\begin{pmatrix}
L^{-1} & 0\\
-P & L
\end{pmatrix}%
\begin{pmatrix}
U & V\\
-V & U
\end{pmatrix}
\label{iwa1}%
\end{equation}
where $P$ \ and $L$ are given by (\ref{pl1})--(\ref{pl1}), and $U$ and $V$ are
given by%
\begin{equation}
U=(AA^{T}+BB^{T})^{-1/2}A\text{ \ },\text{ \ }V=(AA^{T}+BB^{T})^{-1/2}B.
\label{unixy}%
\end{equation}
Now, $R((B^{2n}(\sqrt{\hbar})=B^{2n}(\sqrt{\hbar})$ by rotational symmetry,
hence the result. The only not quite obvious statement here is that $P=P^{T}$
but this follows from the fact that $V_{P}$ must be symplectic, which is
possible if and only if $P$ is symmetric.
\end{proof}

\subsection{The identification of $\operatorname*{Blob}(n)$ and
$\operatorname*{Gauss}(n)$}

We claim that to each quantum blob $S$ $B^{2n}(z_{0},\sqrt{\hbar})$ we can
associate a Gaussian $\psi_{XY}^{(z_{0}}$, and that this association is a
bijection $\operatorname*{Blob}(n)\longleftrightarrow\operatorname*{Gauss}%
(n)$. it is sufficient to prove the result in the centered case, the general
case follows trivially using phase space translations and Weyl--Heisenberg operators.

\begin{proposition}
The mapping $\Phi_{\operatorname*{Blob}}:\operatorname*{Blob}_{0}%
(n)\longrightarrow\operatorname*{Gauss}_{0}(n)$ defined by%
\[
\Phi_{\operatorname*{Blob}}:V_{-P}M_{L}((B^{2n}(\sqrt{\hbar}))\longmapsto
\widehat{V}_{-P}\widehat{M}_{L^{1/2}}\phi_{0}%
\]
where $L=L^{T}>0$ and $\phi_{0}$ the standard coherent state $P=P^{T}$ and is
a bijection. This bijection is related to the Wigner transform of of
$\psi_{XY}$ by the relation
\end{proposition}

\begin{proof}
Recall (formula (\ref{SB})) that every $S(B^{2n}(\sqrt{\hbar}))\in
\operatorname*{Blob}_{0}(n)$can be written as $V_{-P}M_{L}((B^{2n}(\sqrt
{\hbar}))$ and this in a unique way. The mapping $\Phi_{\operatorname*{Blob}}$
is thus well-defined. \ It is also injective for if $_{-P^{\prime}%
}M_{L^{\prime}}\neq_{-P^{\prime}}M_{L^{\prime}}$ then $\widehat{V}%
_{-P^{\prime}}\widehat{M}_{L^{\prime1/2}}\phi_{0}\neq\widehat{V}%
_{-P}\widehat{M}_{L^{1/2}}\phi_{0}$. There remains to show that $\Phi
_{\operatorname*{Blob}}$ is surjective. In view of formula (\ref{fixy}) that
every $\in\psi_{X,Y}\in\operatorname*{Gauss}_{0}(n)$ can be written
$\psi_{X,Y}=\widehat{V}_{Y}\widehat{M}_{X^{1/2}}\phi_{0}$ hence
\[
\psi_{X,Y}=\Phi_{\operatorname*{Blob}}(V_{-Y}M_{X^{2}}((B^{2n}(\sqrt{\hbar
})).
\]

\end{proof}

A related result is that every Fermi ellipsoid contains a quantum blob.

\begin{proposition}
\label{propinclus}Every Fermi ellipsoid contains a quantum blob. More
precisely, for every $\psi_{X,Y}^{x_{0},p_{0}}\in\operatorname*{Gauss}(n)$ we
have
\begin{equation}
\Phi_{\mathrm{blob}}^{-1}(\psi_{X,Y}^{x_{0},p_{0}})\subset\Phi
_{\mathrm{\operatorname{Fermi}}}^{-1}(\psi_{X,Y}^{x_{0},p_{0}})).
\label{inclus}%
\end{equation}

\end{proposition}

\begin{proof}
It is sufficient to consider the case $z_{0}=0.$ By definition%
\[
\Phi_{\mathrm{\operatorname{Fermi}}}^{-1}(\psi_{X,Y}^{x_{0},p_{0}}%
)=\{z::S^{T}D_{X}Sz\cdot z\leq\hbar\operatorname*{Tr}X\}
\]
where
\begin{equation}
D_{X}=%
\begin{pmatrix}
X & 0\\
0 & X
\end{pmatrix}
\text{ \ },\text{ \ }S=%
\begin{pmatrix}
X^{1/2} & 0\\
X^{-1/2}Y & X^{-1/2}%
\end{pmatrix}
.
\end{equation}
On the other it follows from (\ref{omegablob}) \ that we have%
\[
\Phi_{\mathrm{blob}}^{-1}(\psi_{X,Y}^{x_{0},p_{0}})=\{z::S^{T}Sz\cdot
z\leq\hbar\}
\]
hence thew inclusion (\ref{inclus}) requires that
\[
S^{T}Sz\cdot z\leq\hbar\Longrightarrow S^{T}D_{X}Sz\cdot z\leq\hbar
\operatorname*{Tr}X
\]
that is $S^{T}S\geq(\operatorname*{Tr}X)^{-1}S^{T}D_{X}S$ which clearly holds
since:indeed, this is equivalent to the inequality $(\operatorname*{Tr}%
X)^{-1}D_{X}\leq I_{2n}$ which is trivial.
\end{proof}

\subsection{Quantum blobs and symplectic capacity}

All symplectic capacities agree on quantum blobs since the latter are phase
space ellipsoids:%
\begin{equation}
c(SB^{2n}(z_{0},\sqrt{\hbar}))=c(B^{2n}(z_{0},\sqrt{\hbar}))=\pi\hbar.
\label{cap1}%
\end{equation}
We actually introduced the notion of quantum blob in connection with the
following result, which is a symplectically invariant formulation of the
strong uncertainty principle: Recall first Let $\widehat{\rho}$ be a trace
class operator on $L^{2}(\mathbb{R}^{n})$, with trace 1: $\operatorname*{Tr}%
\widehat{\rho}=1.$ It is a density operator (or density matrix, in physics) if
it is in addition positive semi-definite: $\widehat{\rho}\geq0$ (in which case
it is also self-adjoint). Calculating the covariance matrix $\Sigma$ of
$\widehat{\rho}$ (if defined!) from its Weyl symbol $\rho$ one proves that the
condition $\widehat{\rho}\geq0$ implies that we must have \cite{QHA,Narcow}
$\Sigma+\frac{i\hbar}{2}\geq0$, that is
\begin{equation}
\Sigma+\frac{i\hbar}{2}J\text{ \ \emph{is positive semidefinite.}}
\label{quantum}%
\end{equation}
This condition is necessary not sufficient to ensure the condition
$\widehat{\rho}\geq0$, except when $\rho$ is a Gaussian. It turns out that the
condition (\ref{quantum}) implies --but is mot equivalent to-- the
Robertson--Schr\"{o}dinger inequalities \cite{cogoni,Dutta,QHA}: writing
\[
\Sigma=%
\begin{pmatrix}
\Delta x^{2} & \Delta(x,p)\\
\Delta(p,x) & \Delta p^{2}%
\end{pmatrix}
\text{ \ , \ }\Delta(x,p)=\Delta(p,x)^{T}%
\]
where $\Delta x^{2}=(\Delta x_{j}\Delta x_{k})$ the latter are
\begin{equation}
(\Delta x_{j})_{\widehat{\rho}}^{2}(\Delta p_{j})_{\widehat{\rho}}^{2}%
\geq\Delta(x_{j},p_{j})_{\widehat{\rho}}^{2}+\tfrac{1}{4}\hbar^{2}. \label{RS}%
\end{equation}
Condition (\ref{quantum}) implies that the covariance matrix $\Sigma$ is
positive definite, hence invertible. It follows that we can define the
covariance ellipsoid.%
\[
\Omega_{\Sigma}=\{z:\frac{1}{2}\Sigma^{-1}z\cdot z\leq1\}.
\]
We proved in \cite{go09} (also see \cite{goluPR}) the following result which
links the theory of symplectic capacities to the uncertainty principle:

\begin{theorem}
The quantum condition $\Sigma+\frac{i\hbar}{2}J\geq0$ \ is is equivalent to
\begin{equation}
c(\Omega_{\Sigma})\geq\pi\hbar\label{cosig}%
\end{equation}
for every symplectic capacity $c$;
\end{theorem}

Notice it the inequality $c(\Omega_{\Sigma})\geq\pi\hbar$ thus implies the
Robertson--Schr\"{o}dinger inequalities.

\section{Microlocal Pairs and $\psi_{X,Y}$}

The general notion of microlocal pair was defined by Fefferman in
\cite{Fefferman}. Here we specialize the concept, using an extended version of
polar duality between convex bodies. We introduced the idea in
\cite{Poiuntilliswme} and studied further some related mathematical ideas in
\cite{BSM}.

\subsection{Lagrangian polar duality}

We introduced formally the notion of Lagrangian polar duality in \cite{BSM}.
Recall \cite{Birk,Birkbis} that a Lagrangian plane $\ell$ in $(\mathbb{R}%
^{2n},\omega)$ is a$n$-dimensional subspace $\ell$ of $\mathbb{R}^{2n}$ on
which the symplectic form $\omega$ vanishes identically. We denote by
$\operatorname*{Lag}(n)$ the Lagrangian Grassmannian of $(\mathbb{R}%
^{2n},\omega)$:\textquotedblleft\ We will use the notation $\ell
_{X}=\mathbb{R}_{x}^{n}\times0$ and $\ell_{P}=0\times\mathbb{R}_{p}^{n}$ for
the $x$ and $p$ coordinate planes; obviously $\ell_{X},\ell_{P}\in
\operatorname*{Lag}(n)$ There is a natural continuous transitive action
$\operatorname*{Sp}(n)\times\operatorname*{Lag}(n)\longrightarrow
\operatorname*{Lag}(n);$ moreover $\operatorname*{Sp}(n)$ also acts
transitively on the subset of $\operatorname*{Lag}(n)\times\operatorname*{Lag}%
(n)$ consisting of all pairs $(\ell,\ell^{\prime})$ of transfer Lagrangian
planes (\textit{i.e}. $\ell\cap\ell^{\prime}=0$).

In \cite{BSM} we defined the notion of Lagrangian polarity: let $\ell$ and
$\ell^{\prime}$ be two transversal Lagrangian planes, and $X_{\ell}$ a
centrally symmetric convex body in $\ell$. The Lagrangian polar dual
$X_{\ell^{\prime}}^{\hbar}$ of $X_{\ell}$ in $\ell^{\prime}$ is the centrally
symmetric subset of $\ell^{\prime}$ defined by
\begin{equation}
X_{\ell^{\prime}}^{\hbar}=\{z\in\ell^{\prime}:\omega(z,z^{\prime})\leq
\hbar\text{ \ for all \ }z^{\prime}\in X_{\ell}\}. \label{ozz}%
\end{equation}
In the particular case $\ell=\ell_{X}$ and $\ell^{\prime}=\ell_{P}{}_{p}$ this
reduces, taking $\hbar=1$, to ordinary polar duality familiar from convex
geometry: if $X$ is a convex symmetric body in $\mathbb{R}_{x}^{n}$
\begin{equation}
X^{\hbar}=\{p\in\mathbb{R}_{p}^{n}:p\cdot x\leq\hbar\text{ \ for all \ }x\in
X\}. \label{polar}%
\end{equation}

Following result shows that the general case is reduced to the standard case
using symplectic transformation: Let $S\in\operatorname*{Sp}(n)$ be such that
$(\ell,\ell^{\prime})=S(\ell_{X},\ell_{P})$. Then $X_{\ell^{\prime}}^{\hbar
}\subset\ell^{\prime}$ is given by $X_{\ell^{\prime}}^{\hbar}=S(X^{\hbar})$,
that is
\begin{equation}
(X_{\ell},X_{\ell^{\prime}}^{\hbar})=S(X,X^{\hbar}) \label{sxl}%
\end{equation}
\ where $X^{\hbar}$ is the usual polar dual of $X=S^{-1}(X_{\ell})\subset
\ell_{X}$ (see \cite{BSM}).

We will be particularly concerned with the case where $X_{\ell}$ \ os a
centered ellipsoid carried by a Lagrangian plane $\ell$. Suppose first that
$\ell=\ell_{X}$ and $\ell^{\prime}=\ell_{P}{}_{p}$; then the ellipsoid%
\begin{equation}
X=\{x\in\ell_{X}:Ax\cdot x\leq\hbar\} \label{exc1}%
\end{equation}
($A=A^{T}>0$) has polar dual
\begin{equation}
X^{\hbar}=\{p\in\ell_{P}:A^{-1}p\cdot p\leq\hbar\}. \label{exc2}%
\end{equation}
In particular $B_{X}^{n}(\sqrt{\hbar})^{\hbar}=B_{P}^{n}(\sqrt{\hbar})^{\hbar
}$. In the general case it suffices to use the relation (\ref{sxl}). To handle
the case of ellipsoids carried by a arbitrary Lagrangian planes, we rewrite
$Ax\cdot x\leq\hbar$ as
\[%
\begin{pmatrix}
x & 0
\end{pmatrix}%
\begin{pmatrix}
A & 0\\
0 & 0
\end{pmatrix}%
\begin{pmatrix}
x\\
0
\end{pmatrix}
\leq\hbar.
\]
Setting $\widetilde{A}=%
\begin{pmatrix}
A & 0\\
0 & 0
\end{pmatrix}
$ we have, for $S\in\operatorname*{Sp}(n)$ such that $(\ell,\ell^{\prime
})=S(\ell_{X},\ell_{P})$ we have%
\begin{align}
X_{\ell}  &  =S(X)=\{z=x,p):(S^{T})^{-1}\widetilde{A}S^{-1}z\cdot z\leq
\hbar\}\label{sxh1}\\
X_{\ell^{\prime}}^{\hbar}  &  =S(X^{\hbar})=\{z=x,p):(S^{T})^{-1}%
\widetilde{A^{-1}}S^{-1}z\cdot z\leq\hbar\}. \label{sxh2}%
\end{align}

\subsection{Quantum blobs and microlocal pairs}

We begin with the following essential property of microlocal pairs relating
them to quantum blobs: Recall that the John ellipsoid \cite{Ball} of a convex
body in Euclidean space is the unique ellipsoid with maximum volume contained
in this convex body.

\begin{lemma}
\label{Lemma John}The John ellipsoid of $B_{X}^{n}(\sqrt{\hbar})\times
B_{P}^{n}(\sqrt{\hbar})$ is $B^{2n}(\sqrt{\hbar})$.
\end{lemma}

\begin{proof}
By symmetry considerations the John ellipsoid of $B_{X}^{n}(\sqrt{\hbar
})\times B_{P}^{n}(\sqrt{\hbar})$ must be a centered ball $B^{2n}(R)$ whose
orthogonal projections on $\ell_{X}$ and $\ell_{P}$ must be $B_{X}^{n}%
(\sqrt{\hbar})$ and $B_{P}^{n}(\sqrt{\hbar})$, respectively This implies that
we must have $R=\sqrt{\hbar}$ (see \cite{BSM} for a slightly different proof).
\end{proof}

We denote by $\Pi_{\ell,\ell^{\prime}}$ (resp. $\Pi_{\ell^{\prime},\ell}$) the
projection onto $\ell$ (resp. $\ell^{\prime}$) along $\ell^{\prime}$ /resp.
along $\ell^{\prime}$ /resp. $\ell$). They are related to the orthogonal
projections $\Pi_{X}$ and $\Pi_{P}$ on the coordinate planes $\ell_{X}$ and
$\ell_{P}$ by%
\begin{equation}
\Pi_{\ell,\ell^{\prime}}=S\Pi_{X}S^{-1}\text{ \ and \ }\Pi_{\ell^{\prime}%
,\ell}=S\Pi_{P}S^{-1} \label{projs}%
\end{equation}
when $S\in\operatorname*{Sp}(n)$ is such that $(\ell,\ell^{\prime})=S(\ell
_{X},\ell_{P}).$

\begin{proposition}
\label{Thm2}Let $(X_{\ell},X_{\ell^{\prime}}^{\hbar})\subset(\ell,\ell
^{\prime})$ be a microlocal pair where $X_{\ell}$ and $X_{\ell^{\prime}%
}^{\hbar}$ are centered ellipsoid. The There John ellipsoid of $X_{\ell}\times
X_{\ell^{\prime}}^{\hbar}$ is a quantum $\Omega$ blob whose projections are
$\Pi_{\ell,\ell^{\prime}}\Omega=X_{\ell}$ and $\Pi_{\ell^{\prime},\ell}%
\Omega=X_{\ell^{\prime}}^{\hbar}$.
\end{proposition}

\begin{proof}
The proof we give here considerably simplifies the one we gave in \cite{BSM}.
Consider first the case $(\ell,\ell^{\prime})=(\ell_{X},\ell_{P})$ and let
\begin{align}
X  &  =\{x:Ax\cdot x\leq\hbar\}\text{ \ }=A^{-1}(B_{X}^{n}(\sqrt{\hbar
}))\text{,}\\
\text{ \ }X^{\hbar}  &  =\{p:A^{-1}p\cdot p\leq\hbar\}=A(B_{P}^{n}(\sqrt
{\hbar}))
\end{align}
($A=A^{T}>0$).hence
\begin{equation}
X\times X^{\hbar}=M_{A}(B_{X}^{n}(\sqrt{\hbar})\times B_{P}^{n}(\sqrt{\hbar}))
\end{equation}
(recall that $M_{A}\in\operatorname*{Sp}(n)$). It follows from Lemma
\ref{Lemma John} that the John ellipsoid of $X\times X^{\hbar}$ is the quantum
blob $M_{A}(B^{2n}(\sqrt{\hbar}))$, whose orthogonal projections on $\ell_{X}$
and $\ell_{P}$ are $X$ and $X^{\hbar}$. For the general case, choose
$S\in\operatorname*{Sp}(n)$ such that $(\ell,\ell^{\prime})=S(\ell_{X}%
,\ell_{P})$; we have (formulas (\ref{sxh1})--(\ref{sxh2})) $X_{\ell}=S(X)$ and
$X_{\ell^{\prime}}^{\hbar}=S(X^{\hbar})$ hence the John ellipsoid of $X_{\ell
}\times$ $X_{\ell^{\prime}}^{\hbar}$ is $SM_{A}(B^{2n}(\sqrt{\hbar}))$ which
is again a quantum blob. We have, using the relations (\ref{projs})%
\begin{align*}
\Pi_{\ell,\ell^{\prime}}\left[  SM_{A}(B^{2n}(\sqrt{\hbar}))\right]   &
=S\Pi_{X}\left[  M_{A}(B^{2n}(\sqrt{\hbar}))\right]  =SX\\
\Pi_{\ell^{\prime},\ell}\left[  SM_{A}(B^{2n}(\sqrt{\hbar}))\right]   &
=S\Pi_{P}\left[  M_{A}(B^{2n}(\sqrt{\hbar}))\right]  =SX_{\ell^{\prime}%
}^{\hbar}%
\end{align*}
which concludes the proof.
\end{proof}

\subsection{The correspondence between $\psi_{X,Y}$ and microlocal pairs}

\subsection{Microlocal pairs and symplectic capacities}

Here is a very

\begin{proposition}
Let $(X_{\ell}\times X_{\ell^{\prime}}^{\hbar})=S(X\times X^{\hbar})$,
$S\in\operatorname*{Sp}(n)$, be a pure quasi state. We have%
\begin{equation}
c_{\max}(X_{\ell}\times X_{\ell^{\prime}}^{\hbar})=4\hbar. \label{clh}%
\end{equation}
For a general quasi state \ $(X_{\ell}\times P_{\ell^{\prime}})$,
$X_{\ell^{\prime}}^{\hbar}\subset P_{\ell^{\prime}}$ we have
\begin{equation}
c_{\max}(X\times P)=4\lambda_{\max}\hbar\label{yaron1}%
\end{equation}
where $\lambda_{\max}\geq1$ is the number
\begin{equation}
\lambda_{\max}=\max\{\lambda>0:\lambda X^{\hbar}\subset P\}.
\end{equation}

\end{proposition}

\begin{proof}
For a detailed proof when $X$ is an ellipsoid see Prop. 3 in \cite{ACHAPOLAR}%
). In the general case one has to use results Ii \cite{Artstein}
Artstein-Avidan \textit{et al}. show that for $\hbar=1$ we have $c_{\max
}(X\times X^{1})=4$. Using the obvious relation $X=\hbar X^{1}$ we have, by
the conformality property of symplectic capacities,%
\begin{align*}
c_{\max}(X\times X^{\hbar})  &  =c_{\max}(\hbar^{1/2}(\hbar^{-1/2}X\times
\hbar^{1/2}X^{1}))\\
&  =\hbar c_{\max}(\hbar^{-1/2}X\times\hbar^{1/2}X^{1})\\
&  =\hbar c_{\max}(X\times X^{\hbar})
\end{align*}
the last equality because $\hbar^{-1/2}X\times\hbar^{1/2}X^{1}=M_{\hbar^{1/2}%
}(X,X^{1})$ with $M_{\hbar^{1/2}}\in\operatorname*{Sp}(n)$; hence $c_{\max
}(X\times X^{\hbar})=$ $4\hbar$. Formula (\ref{clh}) follows since by the
symplectic invariance of symplectic capacities we have
\[
c_{\max}(X_{\ell}\times X_{\ell^{\prime}}^{\hbar})=c_{\max}((X\times X^{\hbar
}))=c_{\max}(X\times X^{\hbar}).
\]
A similar argument using the formula $c_{\max}(X\times P)=4\lambda_{\max}$ in
\cite{Artstein} leads to (\ref{yaron1}).
\end{proof}

Let us denote by $\operatorname{Micro}(n)$ the set of all microlocal pairs
\ $X_{\ell}\times X_{\ell^{\prime}}^{\hbar}$ and their translates
$T(z_{0})(X_{\ell}\times X_{\ell^{\prime}}^{\hbar})=X_{\ell}\times
X_{\ell^{\prime}}^{\hbar}+z_{0}$.

\section*{APPENDIX\ A: The groups $\operatorname*{Sp}(n)$ and
$\operatorname*{Mp}(n)$}

The matrix $J=%
\begin{pmatrix}
0 & I\\
-I & 0
\end{pmatrix}
$ ($0$ and $I$ the $n\times n$ zero and identity matrices) defines the
standard symplectic form on the phase space $\mathbb{R}_{x}^{2n}$ via the
formula $\sigma(z,z^{\prime})=Jz\cdot z^{\prime}=p\cdot x^{\prime}-p^{\prime
}\cdot x$:%
\[
(z,z^{\prime})=Jz\cdot z^{\prime}=(z^{\prime})^{T}Jz.
\]
The standard symplectic group is denoted by $\operatorname*{Sp}(n)$: it is the
multiplicative group of all real $2n\times2n$ matrices $S$ such that
$\sigma(Sz,Sz^{\prime})=\sigma(z,z^{\prime})$ for all $z,z^{\prime}$;
equivalently:%
\[
S\in\operatorname*{Sp}(n)\Longleftrightarrow SJS^{T}=S^{T}JS=J.
\]
The symplectic group is a connected classical Lie group generated by the
matrices $J$ and the matrices%
\[
M_{L}=%
\begin{pmatrix}
L^{-1} & 0\\
0 & L^{T}%
\end{pmatrix}
\text{ \ , \ }V_{-P}=%
\begin{pmatrix}
I^{-1} & 0\\
P & I
\end{pmatrix}
\]
with $\det L\neq0$, $P=P^{T}$. The symplectic group has a double covering
faithfully represented by a group of unitary operators acting on
$L^{2}(\mathbb{R}^{n})$, the metaplectic group $\operatorname*{Mp}(n).$ It is
generated by the operators $\widehat{J},\widehat{V}_{P},\widehat{M}_{L.m}$
described in the table below, together with their projections $\pi
^{\operatorname*{Mp}}:$ $\operatorname*{Mp}(n)\longrightarrow
\operatorname*{Sp}(n)$

\begin{center}%
\begin{tabular}
[c]{|l|l|l|}\hline
$\widehat{J}\psi(x)=\left(  \tfrac{1}{2\pi i\hbar}\right)  ^{n/2}\int
e^{-\frac{1}{\hbar}x\cdot x^{\prime}}\psi(x^{\prime})d^{n}x^{\prime}$ &
$\overset{\pi^{\operatorname*{Mp}}}{\longrightarrow}$ & $J=%
\begin{pmatrix}
0 & I\\
-I & 0
\end{pmatrix}
$\\\hline
$\widehat{V}_{P}\psi(x)=e^{-\frac{i}{2\hbar}Px\cdot x}\psi(x)$ &
$\overset{\pi^{\operatorname*{Mp}}}{\longrightarrow}$ & $V_{P}=%
\begin{pmatrix}
I & 0\\
-P & I
\end{pmatrix}
$\\\hline
$\widehat{M}_{L.m}\psi(x)=i^{m}\sqrt{|\det L|}\psi(Lx)$ & $\overset{\pi
^{\operatorname*{Mp}}}{\longrightarrow}$ & $M_{L}=%
\begin{pmatrix}
L^{-1} & 0\\
0 & L^{T}%
\end{pmatrix}
$\\\hline
\end{tabular}

\end{center}

(in the last row, the integer $m$ corresponds to a choice of $\arg\det L$).

\section*{APPENDIX\ B: Symplectic Capacities}

n \textit{intrinsic} symplectic capacity assigns a non-negative number (or
$+\infty$) $c(\Omega)$ to every subset $\Omega$ of phase space $\mathbb{R}%
^{2n}$; this assignment is subjected to the following properties:

\begin{itemize}
\item \textbf{Monotonicity:} If $\Omega\subset\Omega^{\prime}$ then
$c(\Omega)\leq c(\Omega^{\prime})$;

\item \textbf{Symplectic invariance:} If $f$ is a symplectomorphism (linear,
or not) then $c(f(\Omega))=c(\Omega)$;

\item \textbf{Conformality:} If $\lambda$ is a real number then $c(\lambda
\Omega)=\lambda^{2}c(\Omega)$;

\item \textbf{Normalization:} We have
\[
c(B^{2n}(R))=\pi R^{2}=c(Z_{j}^{2n}(R));
\]

\end{itemize}

Let $c$ be a symplectic capacity on the phase plane $\mathbb{R}^{2}$. We have
$c(\Omega)=\operatorname*{Area}(\Omega)$ when $\Omega$ is a connected and
simply connected surface. In the general case there exist infinitely many
intrinsic symplectic capacities, but they all agree on phase space ellipsoids
as we will see below. The smallest symplectic capacity is denoted by $c_{\min
}$ (\textquotedblleft Gromov width\textquotedblright): by definition $c_{\min
}(\Omega)$ is the supremum of all numbers $\pi R^{2}$ such that there exists a
canonical transformation such that $f(B^{2n}(R))\subset\Omega$. The fact that
$c_{\min}$ really is a symplectic capacity follows from a deep and difficult
topological result, Gromov's \cite{gr85} symplectic non-squeezing theorem,
alias the principle of the symplectic camel. (For a discussion of Gromov's
theorem from the point of view of Physics see de Gosson \cite{go09}, de Gosson
and Luef \cite{goluPR}.) Another useful example is provided by the
Hofer--Zehnder \cite{HZ} capacity $c^{\mathrm{HZ}}$. It has the property that
it is given by the integral of the action form $pdx=p_{1}dx_{1}+\cdot
\cdot\cdot+p_{n}dx_{n}$ along a certain curve:%
\begin{equation}
c^{\text{HZ}}(\Omega)=\oint\nolimits_{\gamma_{\min}}pdx \label{chz}%
\end{equation}
when $\Omega$ is a compact convex set in phase space; here $\gamma_{\min}$ is
the shortest (positively oriented) Hamiltonian periodic orbit carried by the
boundary $\partial\Omega$ of $\Omega$. This formula agrees with the usual
notion of area in the case $n=1$.

It turns out that all intrinsic symplectic capacities agree on phase space
ellipsoids, and are calculated as follows (see e.g. \cite{Birk,goluPR,HZ}).
Let $M$ be a $2n\times2n$ positive-definite matrix $M$ and consider the
ellipsoid:%
\begin{equation}
\Omega_{M,z_{0}}:M(z-z_{0})^{2}\leq1. \label{ellipsoid}%
\end{equation}
Then, for every intrinsic symplectic capacity $c$ we have
\begin{equation}
c(\Omega_{M,z_{0}})=\pi/\lambda_{\max}^{\sigma} \label{capellipse}%
\end{equation}
where $\lambda_{\max}^{\sigma}=$ is the largest symplectic eigenvalue of $M$.
The symplectic eigenvalues of a positive definite matrix are defined as
follows: the matrix $JM$ ($J$ the standard symplectic matrix) is equivalent to
the antisymmetric matrix $M^{1/2}JM^{1/2}$ hence its $2n$ eigenvalues are of
the type $\pm i\lambda_{1}^{\sigma},..,$ $\pm i\lambda_{n}^{\sigma}$ where
$\lambda_{j}^{\sigma}>0$. The positive numbers $\lambda_{1}^{\sigma}\geq,..,$
$\geq\lambda_{n}^{\sigma}$ are called the \emph{symplectic eigenvalues} of the
matrix $M$.

The definition of an extrinsic symplectic capacity is similar to that of an
intrinsic capacity, replacing the normalization condition with a weaker one:

\begin{itemize}
\item \textbf{Nontriviality:} $c(B^{2n}(R))<+\infty$ and $c(Z_{j}%
^{2n}(R))<+\infty$.
\end{itemize}

In \cite{EH} Ekeland and Hofer defined a sequence $c_{1}^{\mathrm{EH}}$,
$c_{2}^{\mathrm{EH}},...,c_{k}^{\mathrm{EH}},...$ of extrinsic symplectic
capacities having the monotonicity properties%
\begin{equation}
c_{k}^{\mathrm{EH}}(B^{2n}(R))=\left[  \frac{k+n-1}{n}\right]  \pi
R^{2}\ \ \text{,}\ \ c_{k}^{\mathrm{EH}}(Z_{j}^{2n}(R))=k\pi R^{2}. \label{eh}%
\end{equation}
Of course $c_{1}^{\mathrm{EH}}$ is an intrinsic capacity; in fact it coincides
with the Hofer--Zehnder capacity on bounded and convex sets $\Omega$. We have%
\begin{equation}
c_{1}^{\text{EH}}(\Omega)\leq c_{2}^{\text{EH}}(\Omega)\leq\cdot\cdot\cdot\leq
c_{k}^{\text{EH}}(\Omega)\leq\cdot\cdot\cdot\tag{A8}%
\end{equation}
The Ekeland--Hofer capacities have the property that for each $k$ there exists
an integer $N\geq0$ and a closed characteristic $\gamma$ of $\partial\Omega$
such that%
\begin{equation}
c_{k}^{\text{EH}}(\Omega)=N\left\vert \oint\nolimits_{\gamma}pdx\right\vert
\label{acspec}%
\end{equation}
(in other words, $c_{k}^{\text{EH}}(\Omega)$ is a value of the \textit{action
spectrum} of $\partial\Omega$); this formula shows that $c_{k}^{\text{EH}%
}(\Omega)$ is solely determined by the boundary of $\Omega$; therefore the
notation $c_{k}^{\text{EH}}(\partial\Omega)$ is sometimes used in the
literature. The Ekeland--Hofer capacities $c_{k}^{\text{EH}}$ allow us to
classify phase-space ellipsoids. In fact, the non-decreasing sequence of
numbers $c_{k}^{\text{EH}}(\Omega_{M})$ is determined as follows for an
ellipsoid $\Omega:Mz\cdot z\leq1$ ($M$ symmetric and positive-definite): let
\ $(\lambda_{1}^{\sigma},...,\lambda_{n}^{\sigma})$ be the symplectic
eigenvalues of $M$; then
\begin{equation}
\{c_{k}^{\text{EH}}(\Omega):k=1,2,...\}=\{N\pi\lambda_{j}^{\sigma
}:j=1,...,n;N=0,1,2,...\}. \label{cehc}%
\end{equation}
Equivalently, the increasing sequence $c_{1}^{\text{EH}}(\Omega)\leq
c_{2}^{\text{EH}}(\Omega)\leq\cdot\cdot\cdot$ is obtained by writing the
numbers $N\pi\lambda_{j}^{\sigma}$ in increasing order with repetitions if a
number occurs more than once.

\begin{acknowledgement}
This work has been financed by the Austrian Research Foundation FWF (Grant
number PAT 2056623).
\end{acknowledgement}

\textbf{Keywords}: Gewneralized coherent states, quantum blobs, uncertainty
principle; symplectic geometry

\textbf{DATA\ AVAILABILITY\ STATEMENT}: no data has been used created, other
that the LaTex source file

\textbf{CONFLICT\ OF\ INTERESTS}: there are no conflict of interests


\begin{thebibliography}{99}                                                                                               %


\bibitem {Arnold}V. Arnol'd. \textit{Mathematical Methods of Classical
Mechanics}. Graduate Texts in Mathematics, 2nd edition, Springer-Verlag (1989)

\bibitem {Artstein}S. Artstein-Avidan, S., R. Karasev, and Y. Ostrover. From
Symplectic Measurements to the Mahler Conjecture. arXiv:1303.4197 [math.MG]
and arXiv:1303.4197v2 [math.MG] (2013)

\bibitem {Ball}K. M. Ball. Ellipsoids of maximal volume in convex bodies.
\textit{Geom. Dedicata.} 41(2), 241--250 (1992)\textit{ }

\bibitem {best}G. Benenti and G. Strini. Quantum mechanics in phase space:
first order comparison between the Wigner and the Fermi function, \textit{Eur.
Phys. J.} D 57, 117--121 (2010)

\bibitem {iwa}M. Benzi and N. Razouk. On the Iwasawa decomposition of a
symplectic matrix.\textit{ Applied Mathematics Letters} 20, 260--265 (2007)

\bibitem {cogoni}E. Cordero, M, de Gosson, and F. Nicola, On the Positivity of
Trace Class Operators, \textit{Advances in Theoretical and Mathematical
Physics} 23(8), 2061--2091 (2019)

\bibitem {DeGoHi}G. Dennis, M. de Gosson, and B. Hiley. Bohm's Quantum
Potential as an Internal Energy. \textit{Phys. Lett. A} 379(18) 1224--1227 (2015)

\bibitem {Dutta}B. Dutta, N. Mukunda, and R. Simon. The real symplectic groups
in quantum mechanics and optics, \textit{Pramana} 45(6), 471--497 (1995)

\bibitem {EH}Ekeland, I., Hofer, H.: Symplectic topology and Hamiltonian
dynamics, II. \textit{Math. Zeit.} 203, 553--567 (1990)

\bibitem {Fefferman}C. Fefferman. The Uncertainty Principle. \textit{Bull.
Amer. Math. Soc. }9(2), 129--205\textit{ (1983)}

\bibitem {Fermi}E. Fermi. \textit{Rend. Lincei 11, }980 (1930); reprinted in
\textit{Nuovo Cimento 7}, 361 (1930)

\bibitem {Folland}G. B. Folland, \textit{Harmonic Analysis in Phase space},
Annals of Mathematics studies, Princeton University Press, Princeton, N.J. 1989.

\bibitem {1}R. J. Glauber. Coherent and Incoherent States of the Radiation
Field, \textit{Phys. Rev.} 131, 2766 (1963)

\bibitem {BSM}M. de Gosson. Polar Duality Between Pairs of Transversal
Lagrangian Planes; Applications to Uncertainty Principles.\textit{ Bull. sci.
math} 179, 103171 (2022)

\bibitem {QHA}M. de Gosson, \textit{Quantum Harmonic Analysis, an
Introduction}, De Gruyter, 2021

\bibitem {chalk}M. de Gosson, Symplectic Coarse-Grained Dynamics: Chalkboard
Motion in Classical and Quantum Mechanics; \textit{Adv. Theor. Math. Phys.}
24(4) (2020)

\bibitem {Wigner}{\normalsize M. de Gosson. \textit{The Wigner Transform},
World Scientific, Series: Advanced Texts in mathematics, 2017}

\bibitem {blob}M. de Gosson. Quantum blobs. \textit{Found. Phys.} 43\textbf{
}(4), 440--457\textbf{\ }(2013)

\bibitem {Birkbis}M. de Gosson. \textit{Symplectic Methods in Harmonic
Analysis and in Mathematical Physics}. Birkh\"{a}user, Basel, 2011

\bibitem {ACHAPOLAR}M. de Gosson, Two Geometric Interpretations of the
Multidimensional Hardy Uncertainty Principle. \textit{Appl. Comput. Harmon.
Anal.} 42(1), 143--153\textit{\ }(2017)

\bibitem {Birk}M. de Gosson. \textit{Symplectic Geometry and Quantum
Mechanics}, Birkh\"{a}user, Basel (2006)

\bibitem {ICP}M. A.. de Gosson. \textit{The Principles of Newtonian and
Quantum Mechanics}: the need for Planck's constant $h$, with a Foreword by B.
Hiley\textit{. }Imperial College Press, London, 2001

\bibitem {yale}M. de Gosson. The quantum motion of half-densities and the
derivation of Schr\"{o}dinger's equation. \textit{J. Phys. A:Math. Gen.} 31(2) (1998)

\bibitem {go09}M. de Gosson. The Symplectic Camel and the Uncertainty
Principle: The Tip of an Iceberg? \textit{Found. Phys}. 99, 194 (2009)

\bibitem {Poiuntilliswme}M. de Gosson and C. de Gosson. Pointillisme \`{a} la
Signac and Construction of a Quantum Fiber Bundle Over Convex Bodies.
\textit{Found. Physics} 3, 43 (2023)

\bibitem {gohi11}M. de Gosson and B. Hiley. Imprints of the Quantum World in
Classical Mechanics. \textit{Foundations of Physics}, 41(9), 1415--1436 (2011)

\bibitem {goluPR}M. de Gosson and F. Luef. Symplectic Capacities and the
Geometry of Uncertainty: the Irruption of Symplectic Topology in Classical and
Quantum Mechanics. \textit{Phys. Reps.} 484, 131--179 (2009)

\bibitem {gr85}M. Gromov. Pseudoholomorphic curves in symplectic manifolds.
\textit{Inv. Math.} 82(2), 307--347 (1985)

\bibitem {HZ}H. Hofer and H. E. Zehnder. S\textit{ymplectic Invariants and
Hamiltonian Dynamics}. Birkh\"{a}user Advanced texts (Basler Lehrb\"{u}cher)
Birkh\"{a}user Verlag, 1994

\bibitem {3}J. R. Klauder and \textit{B. }Skagerstam\textit{, Coherent
States,} World Scientific, 1985

\bibitem {Littlejohn}R. G. Littlejohn. The semiclassical evolution of wave
packets, \textit{Phys. Reps}. 138(4--5) 193--291 (1986)

\bibitem {Narcow}F. J. Narcowich, Geometry and uncertainty, \textit{J.\ Math.
Phys.} 31(2) (1990)

\bibitem {2}A. Perelomov. \textit{Generalized Coherent States and Their
Applications}, Springer, 1986

\bibitem {Polter}L. Polterovich. \textit{The Geometry of the Group of
Symplectic Diffeomorphisms}, Lectures in Mathematics, Birkh\"{a}user, 2001

\bibitem {Schr1}E. Schr\"{o}dinger. Der stetige \~{A}\oe bergang von der
Mikro- zur Makromechanik, \textit{Naturwissenschaften }14\textit{ }(1926)

\bibitem {Wal}M. Walschaers. Non-Gaussian Quantum States and Where to Find
Them. \textit{PRX Quantum} 2, 030204

\bibitem {zhang}F. Zhang. \textit{The Schur Complement and its Applications},
Springer, Berlin, 2005.%
\[
\]

\end{thebibliography}
\end{document}